\documentclass[a4paper,12pt]{article}
\usepackage{amsmath,amssymb,amsthm}
\usepackage[sort,authoryear]{natbib}

\makeatletter \oddsidemargin-.25in \evensidemargin-.25in
\makeatother

\usepackage{graphicx}

\newtheorem{proposition}{Proposition}
\newtheorem{theorem}{Theorem}
\newtheorem{corollary}{Corollary}
\begin{document}

\date{}
\title{\Large \textbf{Heterogeneous Susceptibles--Infectives model: Mechanistic derivation of\\ the power law transmission function } }

\author{A. Novozhilov\footnote{e-mail: anovozhilov@gmail.com}\\[2mm]
{\small \textit{National Institutes of Health, 8600 Rockville
Pike, Bethesda, MD 20894 USA}}} \maketitle

\begin{abstract}In many epidemiological models a nonlinear transmission
function is used in the form of power law relationship. It is
constantly argued that such form reflects population heterogeneities
including differences in the mixing pattern, susceptibility, and
spatial patchiness, although the function itself is considered
phenomenological. Comparison with large-scale simulations show that
models with this transmission function accurately approximate data
from highly heterogeneous sources. In this note we provide a
mechanistic derivation of the power law transmission function,
starting with a simple heterogeneous susceptibles--infectives (SI)
model, which is based on a standard mass action assumption. We also
consider the simplest SI model with separable mixing and compare our
results with known results from the literature.
\end{abstract}
\section{Introduction}
It is customary to consider transmission function $T(S,I)$, which
describes the incidence rate, i.e., the number of new cases per
time unit, as a main component of any epidemiological model
\cite{Diekmann2000,mccallum2001}. Here we use usual notations for
susceptible and infective individuals denoting them as $S$ and $I$
respectively. Assuming that there is no influx of susceptible
hosts in out model, we can write that
$$
\frac{d}{dt}S(t)=-T(S,I).
$$
Historically the earliest form of the transmission function was a
simple bilinear form, i.e., $T(S,I)=\beta SI$ \cite{Kermack1927},
which follows from the assumptions of random contacts, host
homogeneity, and application of the law of mass action, thereby
implying that the contact rate of any individual is a linear
function of the population size (see \cite{Diekmann1993} for more
details). Here $\beta>0$ is the transmission coefficient. Under the
proportional mixing assumption (the contact rate is fixed), the
transmission function takes the form $T(S,I)=\beta SI/N$, $N$ is a
population size. If the model includes an assumption of the constant
population size these two transmission functions are virtually the
same from any practical viewpoint, whereas variable population size
can yield dramatically different behaviors (e.g.,
\cite{berezovsky2004sem,Berezovskaya2007b,Novozhilov2006a}).

It was early acknowledged that other than bilinear or proportional
mixing transmission functions should be used in epidemiological
models to provide better fit of the model solutions to empirical
data (see \cite{mccallum2001} for a general account of different
models for transmission functions).

One of the most widely used functions has the following form:
\begin{equation}\label{eq1:1}
    T(S,I)=\beta S^{p}I^{q},\,\quad p,\,q>0.
\end{equation}
We will term this transmission function as \textit{power law
relationship}. It was first used in
\cite{wilson1945lma,wilson1945lmab} in the form $T(S,I)=\beta
S^{p} I$ ``to investigate the consequences of various assumptions
when the laws are not known''. Severo \cite{severo1967} considered
general form \eqref{eq1:1} where both $p$ and $q$ are not equal to
one, though he did not give a detailed analysis of the model. Liu
et al. \cite{Liu1987,Liu1986a} gave a thorough analysis of
different compartmental epidemiological models with \eqref{eq1:1}
and showed that incorporating power law transmission function
yields various dynamical behaviors not observable in models with
bilinear incidence rate, e.g., limit cycles and multiple
equilibrium points. Additional analysis and details of such models
can be found in \cite{hethcote1991sem,Hochberg1991}.

Since the first use of the power law transmission function its form
was explained on a basis of ``intrinsic heterogeneity in mixing
pattern'' of a population under question. The exponents $p$ and $q$
were dubbed as ``heterogeneity parameters'' \cite{severo1967}, but
the model itself is considered phenomenological and lacking
mechanical derivation \cite{mccallum2001} in contract to, e.g.,
bilinear relationship, which is based on a dubious but well
established law of mass action \cite{Heesterbeek2005}.

The link between phenomenological power law incidence rate and
population heterogeneity was made explicit when it was shown that
such mean-field models can provide an accurate approximation to
network based simulations that include variation in the strength,
duration, and number of contacts per person. In \cite{Stroud2006}
the transmission function was used in the form $T(S,I)=\beta S^p I$,
whereas full non-linear transmission function was implemented in
\cite{roy2006}. In both cases it was shown that power law
relationship improves the accuracy of mean-field model predictions
when compared with models with bilinear transmission function (see
also \cite{Bansal2007} for a review on comparison of homogeneous and
heterogeneous models).

In this note we show that power law transmission function can be not
only postulated but also derived, using a simple heterogeneous SI
model. The paper organized as follows. In the next section we
formulate a mechanistic heterogeneous SI model from the first
principles. Section 3 gives a brief exposition of necessary
analytical tools. In Section 4 we present the main results of the
study showing that our heterogeneous model is equivalent to a
homogeneous one, but with a non-linear transmission function.

\section{\large Model formulation}
Heterogeneity profoundly affects the dynamics of infection.
Differences in contact rates, spatial distributions of susceptible
hosts, infectiousness and susceptibility of individuals have a
direct effect on disease dynamics. Here, we specifically look into
heterogeneity in disease parameters (such as susceptibility) do not
touching an important topics of heterogeneity mediated by a
structured variable, such as explicit space or age structure. Our
approach is close to the one given in, e.g.,
\cite{Dushoff1999,Dwyer2002,Veliov2005} (see also
\cite{novozhilov2008sec} for more details).

\paragraph{Model I.} We start with a generic assumption that the subpopulation of
susceptible hosts is heterogeneous, and denote $s(t,\omega)$ the
density of susceptibles at time $t$ having parameter value
$\omega$, which determines susceptibility to a particular disease
and varies from individual to individual. The total size of the
susceptibles is given by $S(t)=\int_{\Omega}s(t,\omega)\,d\omega$,
where $\Omega$ is the set of parameter values. Assuming that the
subpopulation of the infectives is homogeneous (later we relax
this assumption), the contact process is described with the law of
mass action, and the rate of change in the susceptibles is
determined by transmission parameter, which is a function of
$\omega$, we obtain that
\begin{equation}\label{eq2:1}
\frac{\partial}{\partial
t}s(t,\omega)=-\beta(\omega)s(t,\omega)I(t).
\end{equation}
Here $\beta(\omega)$ incorporates information on the contact rate
and the probability of a successful contact.

The change in the infective class is given by
\begin{equation}\label{eq2:2}
    \frac{d}{dt}I(t)=I(t)\int_{\Omega}s(t,\omega)\,d\omega=\bar{\beta}(t)S(t)I(t),
\end{equation}
where we denote
$$
\bar{\beta}(t)\int_{\Omega}p_s(t,\omega)\,d\omega,\quad
p_s(t,\omega)=\frac{s(t,\omega)}{S(t)}\,.
$$
Therefore, $\bar{\beta}(t)$ is the mean value of the function
$\beta(\omega)$ with respect to probability density function
$p_s(t,\omega)$ for any time $t$. We need the initial conditions for
the model \eqref{eq2:1},\eqref{eq2:2}:
\begin{equation}\label{eq2:3}
    s(0,\omega)=s_0(\omega)=S_0p_s(0,\omega),\,I(0)=I_0.
\end{equation}
Here $S_0,\,I_0$ are given numbers, and $p_s(0,\omega)$ is a given
initial distribution of the susceptibility in the population.

We note that formally, after integrating equation \eqref{eq2:1}
with respect to $\omega$, we obtain a homogeneous SI model with
non-constant transmission parameter $\bar{\beta}(t)$ which, in its
turn, depends on the current distribution of susceptibility in the
population. If $\bar{\beta}(t)$ is known then the problem is
solved. Interesting to remark that \textit{ad hoc} approach to use
time-dependent transmission coefficient $\beta(t)$ in an SIR model
was used to approximate a heterogeneous epidemics with a
mean-field model \cite{keeling2005ins}.

\paragraph{Model II.} Let us assume now that not only the susceptibles are heterogeneous for
some trait that influences the disease evolution, but also the
infectives are heterogeneous, and consider the simplest possible SI
model. Let $s(t,\omega_1)$ and $i(t,\omega_2)$ be the densities of
the susceptibles and infectives respectively, here we assume that
the traits of the two classes are independent, i.e.,
$\beta(\omega_1,\omega_2)=\beta_1(\omega_1)\beta_2(\omega_2)$. The
number of susceptibles with the trait value $\omega_1$ infected by
individuals with trait value $\omega_2$ is given by
$\beta_1(\omega_1)s(t,\omega_1)\beta_2(\omega_2)i(t,\omega_2)$, and
the total change in the infective class with trait value $\omega_2$
is
$\beta_2(\omega_2)i(t,\omega_2)\int_{\Omega_1}\beta_1(\omega_1)s(t,\omega_1)\,d\omega_1$;
an analogous expression applies to the change in the susceptible
population. We emphasize that nothing else except for the standard
law of mass action is supposed to formulate the terms for the change
in susceptible and infective subpopulations. Combining the above
assumptions we obtain the following model:
\begin{equation}\label{si1}
    \begin{split}
    \frac{\partial}{\partial t}s(t,\omega_1)  &=-\beta_1(\omega_1)s(t,\omega_1)\int_{\Omega_2}\beta_2(\omega_2) i(t,\omega_2)\, d\omega_2\\
           &=-\beta_1(\omega_1)s(t,\omega_1)\bar{\beta}_2(t)I(t),\\
    \frac{\partial}{\partial t}i(t,\omega_2)  &=\beta_2(\omega_2)i(t,\omega_2)\int_{\Omega_1}\beta_1(\omega_1) s(t,\omega_1)\,
    d\omega_1\\
&=\beta_2(\omega_2)i(t,\omega_2)\bar{\beta}_1(t)S(t).
\end{split}
\end{equation}
Model \eqref{si1} is supplemented with the initial conditions
$s(0,\omega_1)=S_0p_s(0,\omega_1),\,i(0,\omega_2)=I_0p_i(0,\omega_2)$.

In \eqref{si1} it is assumed that if an individual having trait
value $\omega_1$ was infected by an individual with trait value
$\omega_2$ he or she becomes an infective with trait value
$\omega_2$. This is a restrictive assumptions which is necessary to
apply the main theorem from the next section.

The global dynamics of \eqref{si1}, as well as of
\eqref{eq2:1}-\eqref{eq2:3}, is simple and is similar to the
simplest homogeneous SI model.

\paragraph{Model III.} Above we were talking about heterogeneity of
the hosts: whether all susceptible individuals are of the same type
with equal susceptibility, and whether all infectious individuals
have equal ability to infect others. Another aspect of heterogeneity
is the possible heterogeneous social contact network
\cite{Bansal2007}. It is difficult to apply the general theory of
heterogeneous populations (see below) to such models, however, there
is a simple case, for which some results can be obtained.

Let us assume that $n(t,\omega)$ denotes the density of individuals
in the population, which are making $\omega$ contacts on average.
Every individual can be contacted by another individual, which
differs is an average number of contact per individual. This
situation is usually termed as \textit{separable mixing}. If we
denote $r$ the probability of transmission the disease given a
contact, then, the simplest SI-model with separable mixing can be
described by the following system:
\begin{equation}\label{eq2:4}
\begin{split}
\frac{\partial}{\partial t}s(t,w)&=-r\omega s(t,\omega)\frac{\int_{\Omega}\omega i(t,\omega)d\omega}{\int_{\Omega}\omega n_0(\omega)d\omega},\\
\frac{\partial}{\partial t}i(t,w)&=-r\omega
s(t,\omega)\frac{\int_{\Omega}\omega
i(t,\omega)d\omega}{\int_{\Omega}\omega n_0(\omega)d\omega},
\end{split}
\end{equation}
where $s(t,\omega)+i(t,\omega)=n_0(\omega)$ for any $t$, and
$n_0(\omega)$ is a given density which specifies probability density
function of contact distribution. Using the property that
$i(t,\omega)=n_0(\omega)-s(t,\omega)$, we obtain
\begin{equation}\label{eq2:5}
\frac{\partial}{\partial t}s(t,w)=-r\omega
s(t,\omega)\left[1-\frac{\int_{\Omega}\omega
s(t,\omega)d\omega}{\int_{\Omega}\omega n_0(\omega)d\omega}\right].
\end{equation}

Models \eqref{eq2:1}-\eqref{eq2:3}, \eqref{si1}, and \eqref{eq2:5}
are infinite dimensional dynamical systems. The special form of the
models, however, allows us to use well developed tools of the theory
of heterogeneous populations, which are presented in the following
section.

\section{Some facts from the theory of heterogeneous populations}
Here we present some results from the theory of heterogeneous
populations in the form suitable for our goal noting that more
general cases can be analyzed \cite{Karev2005a}. For the proofs we
refer to \cite{karev2006}, where similar models are considered.

Let us assume that there are two interacting populations whose
dynamics depend on trait values $\omega_1$ and $\omega_2$
respectively. The densities are given by $n_1(t,\omega_1)$ and
$n_2(t,\omega_2)$, and the total population sizes
$N_1(t)=\int_{\Omega_1}n_1(t,\omega)\,d\omega_1$ and
$N_2(t)=\int_{\Omega_2}n_2(t,\omega)\,d\omega_2$. Obviously, more
than two populations can be considered, or some populations may be
supposed to be homogeneous. Assume next that the net reproduction
rates of the populations have the specific form which is presented
below:
\begin{equation}\label{s2:1}
    \begin{split}
\frac{\partial}{\partial
    t}n_1(t,\omega_1)&=n_1(t,\omega_1)[f_1(\textbf{v}_1)+\varphi_1(\omega_1)g_1(\textbf{v}_1)],\\
    \frac{\partial}{\partial
    t}n_2(t,\omega_1)&=n_2(t,\omega_2)[f_2(\textbf{v}_2)+\varphi_2(\omega_2)g_2(\textbf{v}_2)],
\end{split}
\end{equation}
where $\textbf{v}_1=(N_1,N_2,\bar{\varphi}_2(t))$,
$\textbf{v}_2=(N_1,N_2,\bar{\varphi}_1(t))$, $\varphi_i(\omega_i)$
are given functions,
$\bar{\varphi}_i(t)=\int_{\Omega_i}\varphi_i(\omega_i)p_i(t,\omega_i)\,d\omega_i$
are the mean values of $\varphi_i(\omega_i)$, and
$p_i(t,\omega_i)=n_i(t,\omega_i)/N_i(t)$ are the corresponding
pdfs, $i=1,2$. We also assume that $\varphi_i(\omega_i)$,
considered as random variables, are independent. The system
\eqref{s2:1} plus the initial conditions
\begin{equation}\label{s2:1a}
n_i(0,\omega_i)=N_i(0)p_i(0,\omega_i),\,i=1,2,
\end{equation}
defines, in general, a complex transformation of densities
$n_i(t,\omega_i)$. An effective approach to analyze models in the
form \eqref{s2:1} was suggested in \cite{Karev2000a} (examples of
model analysis are given in
\cite{Karev2003,Karev2005,karev2006,Novozhilov2004}).

Let us denote
$$
M_i(t,\lambda)=\int_{\Omega_i}e^{\lambda\varphi_i(\omega_i)}p_i(t,\omega_i)\,d\omega_i,\quad
i=1,2,
$$
the moment generating functions (mgfs) of the functions
$\varphi_i(\omega_i)$, $M_i(0,\lambda)$ are the mgfs of the initial
distributions, $i=1,2$, which are given.

Let us introduce auxiliary variables $q_i(t)$ as the solutions of
the differential equations
\begin{equation}\label{s2:2}
    dq_i(t)/dt=g_i(\textbf{v}_i),\quad
    q_i(0)=0,\quad
    i=1,2.
\end{equation}

The following theorem holds
\begin{theorem}\label{th1}
Suppose that $t\in[0,T)$, where $T$ is the maximal value of $t$ such
that \eqref{s2:1}-\eqref{s2:1a} has a unique solution. Then

\emph{(i)} The current means of $\varphi_i(\omega_i),\, i=1,2$, are
determined by the formulas
\begin{equation}\label{s2:3}
    \bar{\varphi}_i(t)=\left.\frac{dM_i(0,\lambda)}{d\lambda}\right|_{\lambda=q_i(t)}\frac{1}{M_i(0,q_i(t))}\,,
\end{equation}
and satisfy the equations
\begin{equation}\label{s2:4}
    \frac{d}{dt}\bar{\varphi}_i(t)=g_i(\textbf{\emph{v}}_i)\sigma_i^2(t),
\end{equation}
where $\sigma_i^2(t)$ are the current variances of
$\varphi_i(t,\omega_i)$, $i=1,2$.

\emph{(ii)} The current population sizes $N_1(t)$ and $N_2(t)$
satisfy the system
\begin{equation}\label{s2:5}
    \frac{d}{dt}N_i(t)=N_i(t)[f_i(\textbf{\emph{v}}_i)+\bar{\varphi}_i(t)g_i(\textbf{\emph{v}}_i)],\quad
    i=1,2.
\end{equation}
\end{theorem}

From Theorem 1 follows that the analysis of model
\eqref{s2:1}-\eqref{s2:1a} is reduced to analysis of ODE system
\eqref{s2:2},\eqref{s2:3},\eqref{s2:5}, the only thing we need to
know is the mgfs of the initial distributions.

Concluding this sections we note that, with obvious notation
changes, models \eqref{eq2:1}-\eqref{eq2:3} and \eqref{si1} fall
into the general framework of the master model \eqref{s2:1}.

\section{Model analysis}
We start with the model \eqref{eq2:1}-\eqref{eq2:3}, which,
according to Theorem 1, can be written in the form
\begin{equation}\label{ep4:1}
    \begin{split}
    \frac{d}{dt} S(t) &=-\bar{\beta}(t)S(t)I(t),\quad S(0)=S_0,\\
    \frac{d}{dt}  I(t) &=\bar{\beta}(t)S(t)I(t),\quad I(0)=I_0,\\
    \frac{d}{dt}  q(t) &=-I(t),\quad q(0)=0,\\
    \bar{\beta}(t)&=\left.\frac{dM(0,\lambda)}{d\lambda}\right|_{\lambda=q(t)}\frac{1}{M(0,q(t))}\,.
\end{split}
\end{equation}
$M(0,\lambda)$ is a given mgf of $p_s(0,\omega)$.

\begin{proposition}\label{pr1}
Model \eqref{ep4:1} is equivalent to the following model:
\begin{equation}\label{s3:2}
    \begin{split}
    \frac{d}{dt}S(t) & =-h(S(t))I(t),\quad S(0)=S_0,\\
    \frac{d}{dt}I(t) & =h(S(t))I(t),\quad I(0)=I_0.
\end{split}
\end{equation}
where
\begin{equation}\label{s3:3}
    h(S)=S_0\left[\left.\frac{dM^{-1}(0,\xi)}{d\xi}\right|_{\xi=S/S_0}\right]^{-1},
\end{equation}
and $M^{-1}(0,\xi)$ is the inverse function to mgf $M(0,\lambda)$.
\end{proposition}
\begin{proof}
The first equation in \eqref{ep4:1} can be rewritten in the form
$$
\frac{1}{S(t)}\frac{d}{dt}S(t)=\bar{\beta}(t)\,\frac{d}{dt}q(t)\,.
$$
$\bar{\beta}(t)$ can be represented as $
\bar{\beta}(t)=\left.\frac{d\ln
M(0,\lambda)}{d\lambda}\right|_{\lambda=q(t)}, $ which gives
$$
\frac{d\ln S(t)}{dt}=\frac{d}{dt}\ln M(0,q(t)),
$$
or, using the initial conditions $S(0)=S_0,\,q(0)=0$,
\begin{equation}\label{s3:4}
    S(t)/S_0=M(0,q(t)),
\end{equation}
which is the first integral to system \eqref{ep4:1}. Knowledge of
a first integral allows to reduce the order of the system by one.
Since $M(0,\lambda)$ is an absolutely monotone function in the
case of nonnegative $\beta(\omega)\geqslant 0$, then it follows
that
\begin{equation}\label{s3:5}
q(t)=M^{-1}\left(0,S(t)/S_0\right),
\end{equation}
where $M^{-1}(0,M(0,\lambda))=\lambda$ for any $\lambda$.

Putting \eqref{s3:5} into \eqref{ep4:1} gives
$$
\frac{d}{dt}S(t)=\left.\frac{dM(0,\lambda)}{d\lambda}\right|_{\lambda=M^{-1}\left(0,S(t)/S_0\right)}S_0I(t),
$$
or, by the inverse function theorem, \eqref{s3:2} with
\eqref{s3:3}.
\end{proof}

Note that model \eqref{si1} can be reduced to four-dimensional
system of ODEs, which, in its turn, can be simplified to
two-dimensional system. The proof is as in Proposition~1.
Formally, we have
\begin{proposition}\label{pr2}
The model \eqref{si1} is equivalent to the model
\begin{equation*}
    \begin{split}
    \frac{d}{dt}S(t) &= -h_1(S)h_2(I),\\
    \frac{d}{dt}I(t) &= h_1(S)h_2(I),
\end{split}
\end{equation*}
where $h_i(x),\,i=1,2$ are given by \eqref{s3:3}.
\end{proposition}

Combining together Propositions 1 and 2 we obtain the main result
of the present note.
\begin{theorem}
A heterogeneous SI model in the form \eqref{eq2:1}-\eqref{eq2:3}, or
in the form \eqref{si1}, which both describe the contact process
with the help of the law of mass action and model heterogeneities in
disease parameters such as susceptibility to a disease or
infectivity of an individual, are equivalent to a homogeneous SI
model with a nonlinear transmission function.
\end{theorem}

An analogous conjecture was made in \cite{Veliov2005}, where a
substantially more complex model is analyzed. The strength of
Theorem 2 is that it provides an explicit form for the nonlinear
transmission function.

Consider a standard gamma distribution with parameters $k$ and
$\nu$:
\begin{equation}\label{gamma}
    p(0,\omega)=\frac{\nu^k}{\Gamma(k)}\omega^{k-1}e^{-\nu\omega},\quad
    \omega\geqslant 0,\,k>0,\,\nu>0.
\end{equation}
Let us assume that $\beta(\omega)=\omega$. The mgf of
gamma-distribution is then
$$M(0,\lambda)=(1-\lambda/\nu)^{-k}.$$ Using Proposition 1 we
obtain that
\begin{equation}\label{gd1}
h(S)=\frac{kS}{\nu}\left[\frac{S}{S_0} \right]^{1/k}.
\end{equation}
From \eqref{gd1} it immediately follows
\begin{corollary}
The power relationship \eqref{eq1:1} with $q=1,p=1+1/k$ can be
obtained as a consequence of the heterogenous SI model with
distributed susceptibility when the initial distribution is a
gamma-distribution with parameters $k$ and $\nu$.
\end{corollary}

\begin{corollary}
The power relationship \eqref{eq1:1} with $q=1+1/k_2,p=1+1/k_1$ can
be obtained as a consequence of the heterogenous SI model with
distributed susceptibility and infectivity when the initial
susceptibility distribution is a gamma-distribution with parameters
$k_1$ and $\nu_1$, and the initial infectivity distribution is a
gamma-distribution with parameters $k_2$ and $\nu_2$.
\end{corollary}

Summarizing, we provided a mechanistic derivation of the power law
transmission function, which was used phenomenologically in many
epidemiological models, in the case when heterogeneity parameters
$p,\,q$ exceed one. Originally, these exponents were considered to
be less than one (e.g., in \cite{severo1967} they are put in the
form $p=1-a,\,q=1-b$), but no comparison with real world data was
provided.

There is no universal agreement on the values of parameters
$p,\,q$ in \eqref{eq1:1}. In \cite{roy2006} these parameters were
estimated when the incidence rate was inferred from epidemic
simulations on random networks with different degree
distributions. In all experiments values of $p$ and $q$ were
estimated to be less than 1. In contrast to the last observation,
in \cite{Stroud2006}, where the transmission function has the form
$T(S,I)=\beta S^pI$, it was argued that the exponent $p$ should be
greater than one. Fitting the solutions of the mean field model
with nonlinear transmission function into the data obtained from
large-scale simulations, it was found that $p$ can range from 1.6
to 2.

In any respect, the question of deriving the power law
transmission function on a solid mechanistic bases for the case
$p,\,q<1$ remains open, whereas the case $p,\,q>1$ is fully
covered by Corollaries 1 and 2.

\section{Model III and separable mixing}
We rewrite equation \eqref{eq2:5} in the form
\begin{equation}\label{eq5:1}
\frac{\partial}{\partial t}s(t,w)=-r\omega
s(t,\omega)\left[1-\frac{\bar{\omega}(t)S(t)}{K}\right],
\end{equation}
where $K$ is the number of contacts, which are made by the total
population, $\bar{\omega}(t)$ is the average number of contacts made
by one susceptible individual at time $t$. We note that formally eq.
\eqref{eq5:1} is not covered by Theorem \ref{th1}, because its
growth coefficient depends on the average parameter value
$\bar{\omega}(t)$. However, it is possible to extend the theory
presented in Section 3 to such cases with minor changes in notations
(Karev, personal communication). In particular, it is possible to
show that equation \eqref{eq5:1} is equivalent to the following
ordinary differential equation:
\begin{equation}\label{eq5:2}
\frac{d}{dt}S(t)=-rh(S)\left[1-\frac{h(S)}{K}\right],
\end{equation}
where $h(S)$ is given by \eqref{s3:3}.

It is interesting to note that we can compare solutions of
\eqref{eq5:1} with solutions of the system of ODEs, obtained as a
result of large mixing rates in the model on dynamic contact network
\cite{volz2007sir}. For SI-model system (2.22)-(2.23) from the cited
work reads
\begin{equation}\label{eq5:3}
    \begin{split}
    \dot{\theta} & =-rM_I\theta,\\
     \dot{M}_I   & =\frac{rM_I}{g'(1)}(\theta
     g'(\theta)+\theta^2g''(\theta)),
\end{split}
\end{equation}
where $g(x)$ is the probability generation function for the
distribution of the number of contacts in the population (this is
PGF for pdf $n_0(\omega)/\int_{\Omega}n_0(\omega)d\omega$);
$\theta(t)$ is the fraction of individuals that have only one
contact and still susceptible by the time $t$; $r$ is the
transmission rate; and $M_I$ is the fraction of contacts made by
infected individuals. The number of susceptible individuals is given
by $S(t)=g(\theta(t))$.

To compare models \eqref{eq5:2} and \eqref{eq5:3} we need to specify
the initial conditions. Since model \eqref{eq5:3} deals with PGF of
the number of contacts of the total population, and eq.
\eqref{eq5:2} incorporates mgf of the number of contacts of
susceptible individuals it is reasonable to expect some discrepancy
of the corresponding solutions if we use the same pdf for these
purposes. See Fig. \ref{f1} for three solutions.

\begin{figure}[tbh!]
\centering
\includegraphics[width=0.9\textwidth]{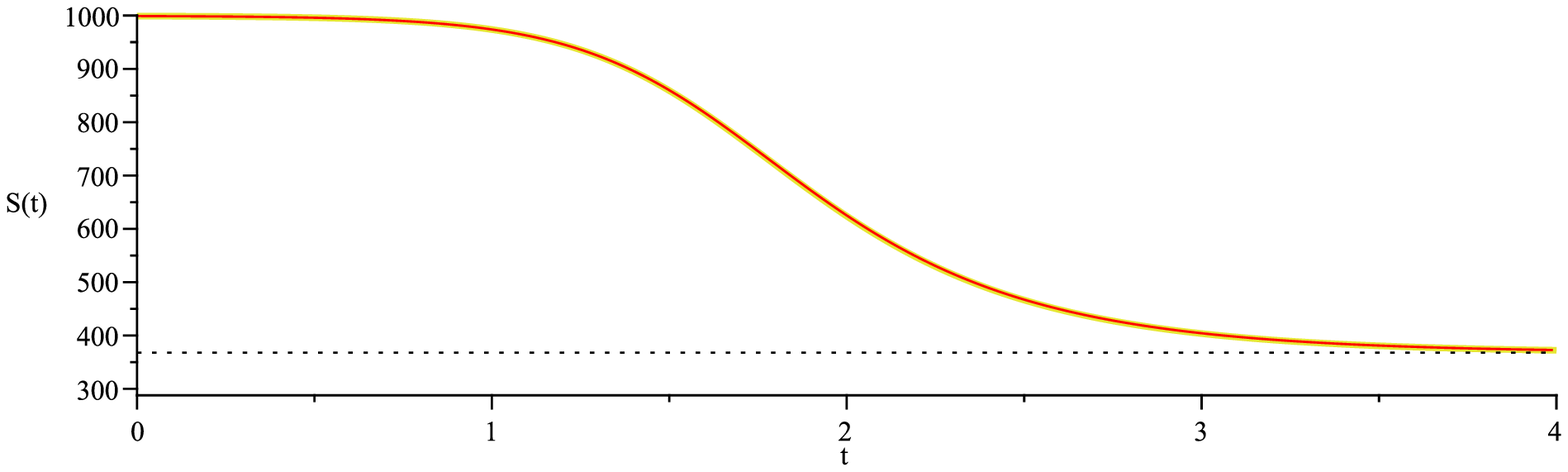}
\includegraphics[width=0.9\textwidth]{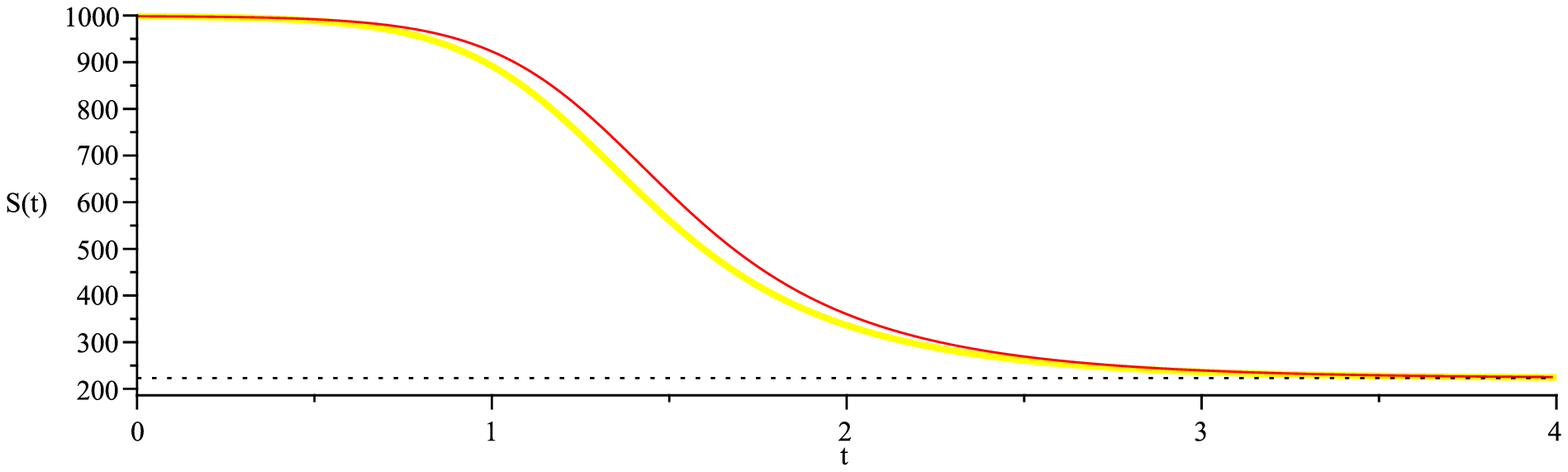}
\includegraphics[width=0.9\textwidth]{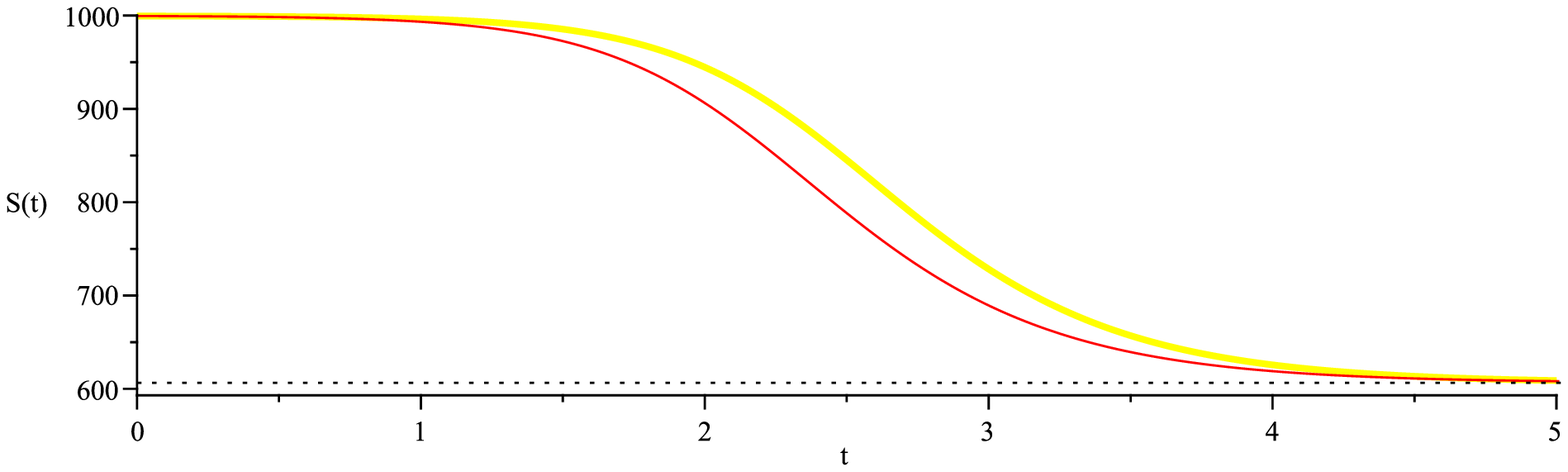}
\caption{Comparison of the solutions of system \eqref{eq5:3} (red
line) with solution of eq. \eqref{eq5:2} (yellow bold line). Poisson
distribution was used with parameters $\lambda=1,\,1.5,\,0.5$ from
top to bottom. $r=2$. It was assumed that the population size is
$N=1000$. The initial conditions for problem \eqref{eq5:3} were
chosen such that $\theta(0)=1-\varepsilon,\,M_I(0)=\varepsilon$,
where $\varepsilon=0.01$. $S(0)$ for \eqref{eq5:2} was found as
$g(1-\varepsilon)$. The dotted line shows $Np_0$, where $p_0$ is the
proportion of individuals in the population who do not make the
contacts}\label{f1}
\end{figure}

As can be seen from Fig. \ref{f1} the best agreement os found when
we use $\lambda=1$, i.e., the average number of contacts equals to
1. In this case two solutions coincide. In the cases $\lambda>1$ or
$\lambda<1$ there is some divergence, although the limiting behavior
of the models is the same.

\end{document}